\documentclass{article}
\usepackage{ijcai22}

\usepackage{times}

\usepackage{soul}
\usepackage{url}
\usepackage[hidelinks]{hyperref}
\usepackage[utf8]{inputenc}
\usepackage[small]{caption}
\usepackage{graphicx}
\usepackage{amsmath}
\usepackage{amssymb}
\usepackage{booktabs}
\usepackage{amsthm}

\usepackage{multirow}
\usepackage{tabularx}

\usepackage{caption, subcaption}

\usepackage[ruled,vlined]{algorithm2e}
\include{pythonlisting}
\newtheorem{theorem}{Theorem}[section]

\SetKwInput{KwInput}{Input}
\SetKwInput{KwOutput}{Output}

\title{Robust Speech Representation Learning via \\ Flow-based Embedding Regularization}

\author{
Woo Hyun Kang, Jahangir Alam, Abderrahim Fathan \\
\affiliations
Computer Research Institute of Montreal \\
\emails
\{woohyun.kang, jahangir.alam, abderrahim.fathan\}@crim.ca
}

\begin{document}

\maketitle

\begin{abstract}
Over the recent years, various deep learning-based methods were proposed for extracting a fixed-dimensional embedding vector from speech signals.
Although the deep learning-based embedding extraction methods have shown good performance in numerous tasks including speaker verification, language identification and anti-spoofing, their performance is limited when it comes to mismatched conditions due to the variability within them unrelated to the main task.
In order to alleviate this problem, we propose a novel training strategy that regularizes the embedding network to have minimum information about the nuisance attributes.
To achieve this, our proposed method directly incorporates the information bottleneck scheme into the training process, where the mutual information is estimated using the main task classifier and an auxiliary normalizing flow network.
The proposed method was evaluated on different speech processing tasks and showed improvement over the standard training strategy in all experimentation.
\end{abstract}

\section{Introduction}

In recent years, attributed to the wide deployment of smart devices, the interest in speech-based applications has been rapidly growing.
One major sub-field that is gaining popularity is speech-based identity recognition task, which includes speaker verification, language identification, and voice spoof detection.
Since the given speech signals are likely to have different durations, usually an utterance-level fixed-dimensional vector (i.e., embedding vector) is extracted and fed into a scoring or classification algorithm.
To achieve this, various methods have been proposed utilizing deep learning architectures for extracting embedding vectors and have shown state-of-the-art performance when a large amount of training data is available \cite{dvec1}, \cite{dvec2}, \cite{xvec1}, \cite{xvec2}, \cite{gvlad}, \cite{resnet1}, \cite{resnet2}.
However, despite their success in well-matched conditions, the deep learning-based embedding methods are vulnerable to the performance degradation caused by mismatched conditions \cite{mismatch}.

Recently, many attempts have been made to extract an embedding vectors robust to variability unrelated to the main task \cite{mismatch}, \cite{grl}, \cite{gan}, \cite{jfe}.
Especially in \cite{jfe}, a joint factor embedding (JFE) technique was proposed where the embedding network is trained to maximize the speaker-dependent information within the embedding while simultaneously maximizing the uncertainty on unwanted attributes (e.g., channel, emotion).
Also in \cite{gan}, an adversarial training strategy is proposed, where the embedding network and a nuisance attribute discriminator network are trained in a competitive fashion.
The authors of \cite{grl} also adopted an adversarial strategy, but used a gradient reversal layer to force the embedding network to learn no information about the unwanted attributes.
Although these methods were able to enhance the verification performance, they can only be used when a training set with nuisance labels is available due to their fully-supervised nature.

In this paper, we propose a novel approach to disentangle the unwanted information from the embedding vector without the need for any nuisance labels.
Our proposed method exploits the information bottleneck framework, where the system is trained to produce a latent variable with maximum information on the main-task (e.g., speaker verification) while regularizing it to have minimum redundant information.
In order to minimize the redundancy, we estimate the upper-bound of the mutual information between the input speech and the embedding using the contrastive log-ratio upper-bound (CLUB) scheme \cite{club}.
Since CLUB requires a conditional likelihood estimator, we adopted a simple normalizing flow model which showed an outstanding performance in image generation and speech synthesis tasks.
To maximize the speaker discriminability of the embedding while minimizing its information on unwanted variability within the input speech, we trained the embedding network and the flow-based conditional likelihood estimator in a competitive fashion, similar to the generative adversarial network (GAN) \cite{gan}.
Experimental results showed that the proposed regularization technique was able to enhance the speaker verification performance by suppressing the nuisance information within the embedding vector.

The contributions of this paper are as follows:
\begin{itemize}
    \item We propose a new method to regularize the embedding network to have minimum redundancy, which can be done without the need of any additional metadata (i.e., nuisance labels).
    \item We compared the proposed regularization method on various tasks including speaker verification, language identification and voice spoof detection.
\end{itemize}

\section{Backgrounds}

\subsection{Deep Speech Embedding}

\begin{figure}
	\centering
	\includegraphics[width=0.95\linewidth]{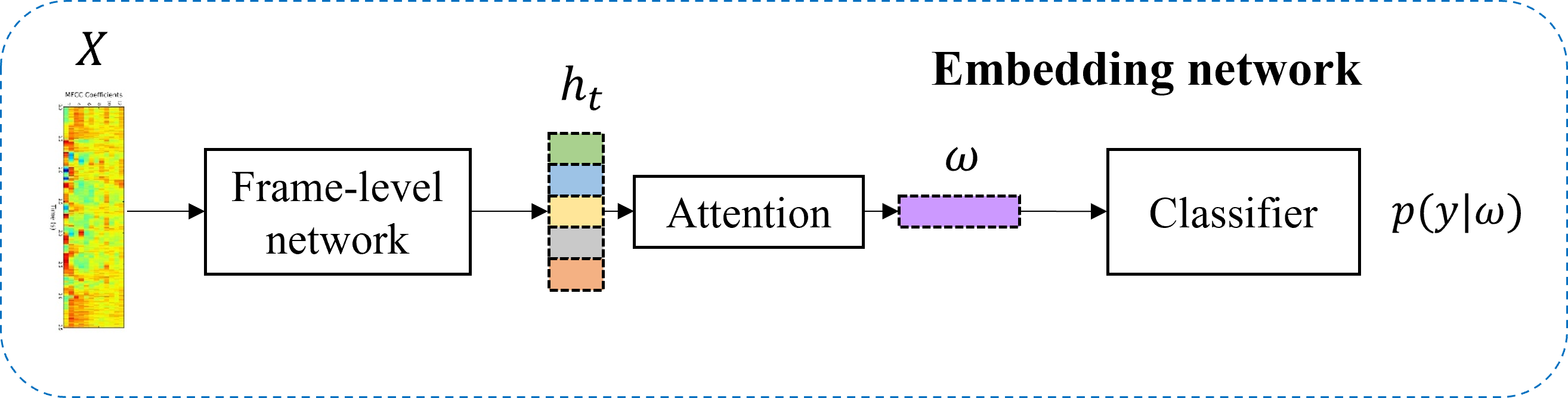}
	\caption{The general architecture of the conventional deep speech embedding systems.}
	\label{emb_net}
\end{figure}

In the past several years, various deep architectures for speech embedding extraction have been proposed.
In most of these frameworks, given a speech utterance $\mathbf{X}$ with $T$ frames, a sequence of frame-level acoustic features $\{\mathbf{x}_{1},...,\mathbf{x}_{T}\}$ extracted from $\mathbf{X}$ is fed into the frame-level network.
Once the frame-level outputs $\{\mathbf{h}_{1},...,\mathbf{h}_{T}\}$ are obtained, they are aggregated to obtain an utterance-level representation.
One way of aggregating the frame-level outputs is the self-attentive pooling (SAP) \cite{sap}, which computes the weighted average as
\begin{equation} \label{attention_1}
\mathbf{\omega}=\sum_{t=1}^{T}{\alpha}_{t}\mathbf{h}_{t}
\end{equation}
where $\alpha_{t}{\in}[0, 1]$ is a normalized weight, which is computed by
\begin{equation} \label{attention_2}
{\alpha}_{t}=\frac{{\exp}(e_t)}{\sum_{t=1}^{T}{\exp}(e_t)}.
\end{equation}
In (\ref{attention_2}), the frame-level score (i.e. attention) $e_t$ is computed as follows:
\begin{equation} \label{attention_3}
e_t=\mathbf{v}^{\intercal}_{t}{\tanh}(\mathbf{W}_{t}\mathbf{h}_{t}+\mathbf{b}_{t})
\end{equation}
where $\mathbf{v}_t$, $\mathbf{W}_t$, and $\mathbf{b}_t$ are trainable parameters and superscript $\intercal$ indicates transpose operation.
By using different weight for each frame, speech frames with relatively higher target-relevancy can contribute more to the embedding vector.

The embedding network is trained to maximize the target discriminability.
Depending on the main task, different types of objective functions are used to achieve this.
For binary classification tasks, such as antispoofing, the system is often trained via one-class softmax objective function, which can be formulated as \cite{zhang2021ocs}:
\begin{equation}
    L_{OCS} = - \frac{1}{N} \sum_{i=1}^N
    log(1+e^{k(m_{y_i} - \hat{W}_0 \hat{\omega}_i) (-1)^{y_i} })
    \label{ocs}
\end{equation}
where $k$ is the scale factor, $\omega_i \in R^D$ and $y_i \in \left\{0,1\right\}$ are the D-dimensional embedding vector and label of the $i^{th}$ sample respectively. 
$N$ is the mini-batch size and $m_{y_i}$ defines the compactness margin for class label $y_i$. 
The larger is the margin, the more compact the embeddings will be. 
$W_0$ is the weight vector of our target class embeddings. Both $\hat{W}_0$ and $\hat{\omega}_i$ are normalizations of $W_0$ and $\omega_i$ respectively.

On the other hand, for multi-class classification tasks, one of the most popular softmax variant objectives is the additive angular margin softmax (AAMSoftmax) \cite{aamsoftmax}.
The AAMSoftmax objective is formulated as follows:
\begin{equation}
    L_{AAMSoftmax} = - \frac{1}{N} \sum_{i=1}^N log(\frac{e^{s(cos(\theta_{y_i,i} + m))}}{K_1}),
    \label{aam}
\end{equation}
where $K_1={e^{s(cos(\theta_{y_i,i} + m))} + \sum_{j=1, j \neq i}^c e^{s cos \theta_{j,i}}}$, $N$ is the batch size, $c$ is the number of classes, $y_i$ corresponds to label index, $\theta_{j,i}$ represents the angle between the column vector of weight matrix $W_j$ and the $i$-th embedding $\omega_i$, where both $W_j$ and $\omega_i$ are normalized. 
The scale factor $s$ is used to make sure the gradient is not too small during the training and $m$ is a hyperparameter that  encourages the similarity of correct classes to be greater than that of incorrect classes by a margin $m$.


\subsection{Information Bottleneck}
The information bottleneck is a theoretical method for learning a latent representation with minimum redundant information.
Given an input data $X$ and its corresponding target label $y$, the information bottleneck aims to learn an encoder that produces a latent variable (embedding) $\omega$ with high information on $y$ and low relevance with the source $X$.
To achieve this, the objective for information bottleneck can be written as:
\begin{equation} \label{ib}
    L_{IB}=-I(y;\omega)+\beta{I(X;\omega)}
\end{equation}
where hyperparameter $\beta$ is a positive scalar coefficient.

Maximizing $I(y;\omega)$ can be easily accomplished by using a typical discriminative objective function including softmax or contrastive losses.
This is mainly attributed to the fact that the softmax and contrastive objective functions can be interpreted as a lower-bound of the mutual information.
Moreover, one could use neural estimation methods such as mutual information neural estimation (MINE) \cite{mine} or InfoNCE \cite{cpc}, which also estimates the lower bound of the mutual information.

On the other hand, minimizing $I(X;\omega)$ can be difficult due to the complex distribution of $X$.
In order to alleviate this problem, most previous approaches regularized the latent distribution to be Gaussian.
However, using such simple distribution for the embedding may hinder its representational ability.
Therefore various attempts have been made to exploit the neural mutual information estimation scheme (e.g., MINE, InfoNCE).
Although these neural estimation methods do not assume the embedding to follow a certain distribution, minimizing these estimations showed minimum improvement since they estimate the lower-bound mutual information.
Thus in order to effectively minimize $I(X;\omega)$, one should estimate the upper-bound mutual information.

\subsubsection{Contrastive Log-ratio Upper-Bound Mutual Information}
CLUB \cite{club} is a mutual information estimation method that is trained via contrastive learning.
Given the conditional distribution $p(X|\omega)$, the mutual information CLUB is defined as:
\begin{equation}
    \begin{aligned}
        I_{CLUB}(X;\omega)=&E_{p(X,\omega)}[\log{p(X|\omega)}]\\
        &-E_{p(X)p(\omega)}[\log{p(X|\omega)}].
    \end{aligned}
    \label{club}
\end{equation}
Unlike the mutual information estimated via MINE or InfoNCE, $I_{CLUB}(X;\omega)$ is the upperbound of the true mutual information $I(X;\omega)$.

\subsection{Normalizing Flow}
Normalizing flow is a generative model which consists of a stack of invertible functions which map the samples from a simple distribution $p_{z}(z)$ to a complex distribution $p_{X}(x)$.
Let $f_i$ be a mapping from $z^{i-1}$ to $z^{i}$, $z^0=x$ and $z^n=z$.
Then $x$ is transformed into $z$ through a chain of invertible mappings:
\begin{equation}
    z=f_{n}{\circ}f_{n-1}{\circ}{\dots}{\circ}f_1(x).
\end{equation}
By the change of variables theorem, the log-likelihood of $x$ can be computed as follows:
\begin{equation}
    \log{p_X(x)}=\log{p_Z(z)}+\sum^{n}_{i=1}\log{|\det{\frac{{\partial}f_i}{{\partial}z^{i-1}}}|}.
\end{equation}
Usually, the latent distribution $p_Z$ is set to be a standard normal distribution $N(0, I)$.
Since the second term $\sum^{n}_{i=1}\log{|\det{\frac{{\partial}f_i}{{\partial}z^{i-1}}}|}$ can be computationally expensive to obtain, $f_i$ is required to have a tractable Jacobian.
One way to achieve this is to use an affine coupling layer which is defined as follows:
\begin{equation}
    z_{a}^{i}=z_{a}^{i-1},
\end{equation}
\begin{equation}
    z_{b}^{i}=z_{b}^{i-1}{\odot}\exp(\sigma(z_{a}^{i-1})+\mu(z_{a}^{i-1})),
    \label{affine_flow}
\end{equation}
where $z_{a}^{i}$ is the first half, $z_{b}^{i}$ is the second half of $z^i$, and $\odot$ is the channel-wise product.
The Jacobian matrix of the affine coupling layer is a lower triangular matrix, which allows efficient computation for $\log{|\det{\frac{{\partial}f_i}{{\partial}z^{i-1}}}|}$:
\begin{equation}
    \log{|\det{\frac{{\partial}f_i}{{\partial}z^{i-1}}}|}=\sum^{D/2}_{j=1}\sigma(z_{a}^{i-1})_j,
\end{equation}
where $\sigma(z_{a}^{i-1})_j$ is the $j^{th}$ element of $\sigma(z_{a}^{i-1})$.

\section{Flow-ER: Flow-based embedding regularization}


\begin{figure}
	\centering
	\includegraphics[width=0.95\linewidth]{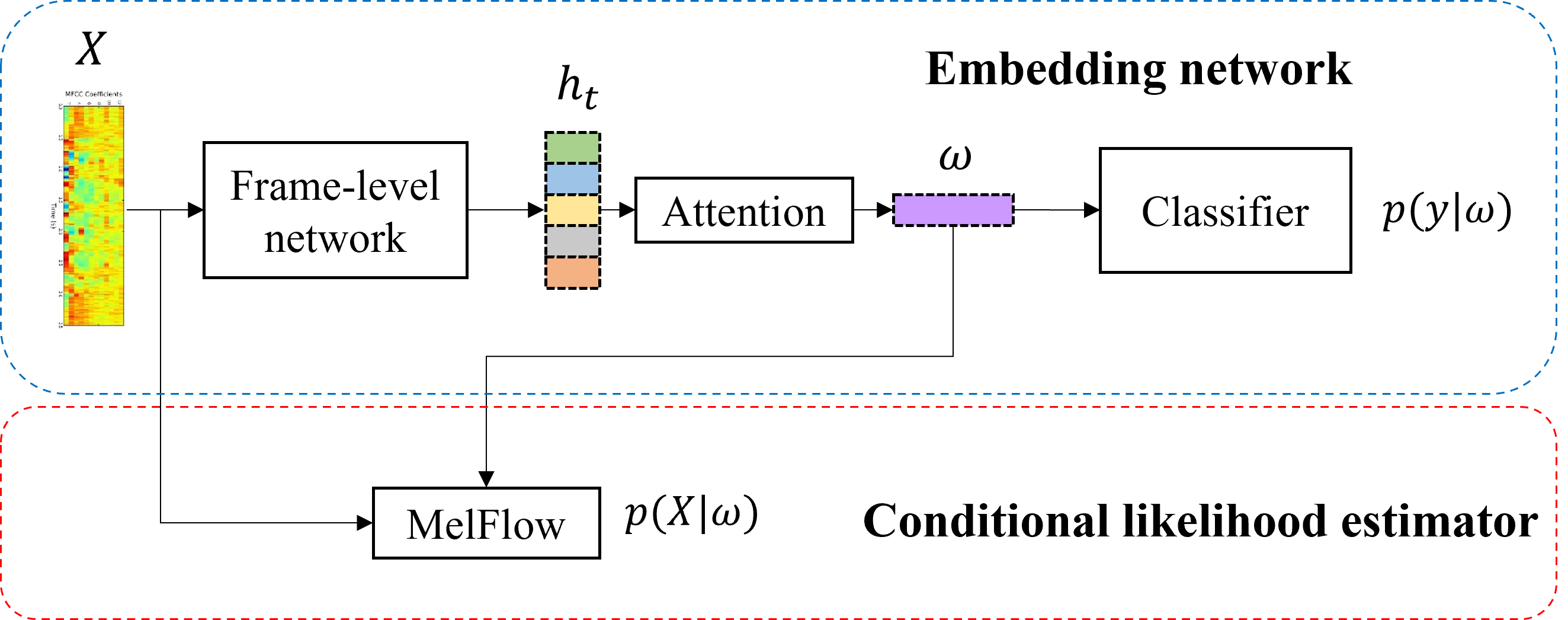}
	\caption{The general architecture of the proposed Flow-ER framework.}
	\label{flower_arch}
\end{figure}

In our proposed method, we aim to extract an embedding $\omega$ from the input speech $X$ with maximum target information while suppressing redundant information latent within $X$ (e.g., channel, noise).
To achieve this, we train the embedding network according to the information bottleneck scheme described in Equation \ref{ib}.

\subsection{Mutual information interpretation of the cross-entropy loss}
In order to maximize $I(y;\omega)$, our system adopts cross-entropy-based objective functions, such as softmax-based losses.
The cross-entropy-based loss functions can be interpreted as the lower bound of the mutual information as in the MINE framework \cite{mine}.
\begin{theorem}
Let $\omega$ and $y$ be random variables with respective priors $p(\omega)$ and $p(y)$, and joint distribution $p(\omega, y)$.
Then for an arbitrary function $g(\omega, y)$, the following holds:
\begin{equation}
L_{xent}=-E_{\omega, y \sim p(\omega, y)}[\log\frac{g(\omega, y)}{\sum_{j=1}^{N}g(\omega, j)}]{\leq} I(y;\omega)
\end{equation}
\end{theorem}
\begin{proof}
Let us define $G(\omega, y)={\log}g(\omega, y)$. 
Then we can re-write $-L_{xent}=E_{\omega, y \sim p(\omega, y)}[\log\frac{g(\omega, y)}{\sum_{j=1}^{N}g(\omega, j)}]$ as:
\begin{equation*}
\begin{split}
    \begin{aligned}
        &-L_{xent}=E_{\omega, y \sim p(\omega, y)}[\log\frac{g(\omega, y)}{\sum_{j=1}^{N}g(\omega, j)}] \\
        &=E_{\omega, y \sim p(\omega, y)}[G(\omega, y)]-E_{\omega \sim p(\omega)}[\log{\sum_{j=1}^{N}g(\omega, j)}] \\
        &{\leq}E_{\omega, y \sim p(\omega, y)}[G(\omega, y)]-{\log}E_{\omega \sim p(\omega)}[{\sum_{j=1}^{N}g(\omega, j)}] \\
        &=E_{\omega, y \sim p(\omega, y)}[G(\omega, y)] -{\log}E_{\omega \sim p(\omega)}[N E_{y \sim p(y)}[g(\omega, j)]] \\
        &=E_{\omega, y \sim p(\omega, y)}[G(\omega, y)] -{\log}E_{\omega, y \sim p(\omega)p(y)}[g(\omega, j)]-{\log}N \\
        &{\leq} I(y;\omega).
    \end{aligned}
    \end{split}
    \label{xent_mi}
\end{equation*}
\end{proof}

Therefore, minimizing the cross-entropy-based loss functions, such as Eq. \ref{ocs} or Eq. \ref{aam} can maximize the mutual information between the embedding vectors $\omega$ and the labels $y$.

\subsection{Mutual information upperbound and Conditional likelihood estimation via Normalizing Flow}

\begin{figure}
	\centering
	\includegraphics[width=0.65\linewidth]{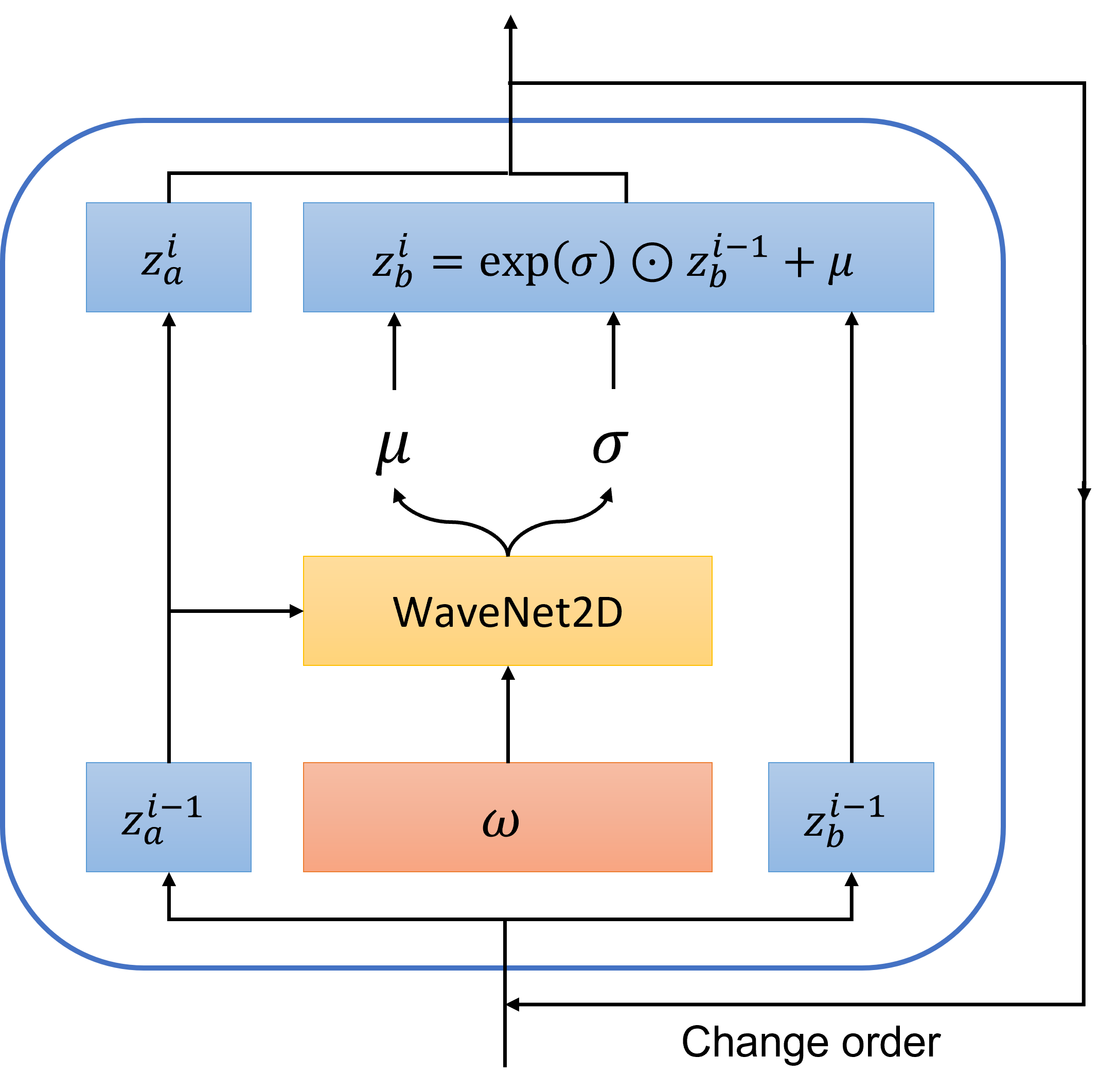}
	\caption{The general architecture of the MelFlow model.}
	\label{melflow_arch}
\end{figure}

To minimize $I(X;\omega)$, we aim to estimate and minimize the upperbound of the mutual information via the CLUB formulation Eq. \ref{club}.
However, in order to achieve this, we need to estimate the conditional likelihood $p(X|\omega)$.
To achieve this, we propose to use a conditional normalizing flow model, similar to the MelFlow proposed in \cite{melflow}.
More specifically, we use a WaveNet2D, a non-causal 2D-convolutional network for computing $\sigma(z)$ and $\mu(z)$ of Eq. \ref{affine_flow}.
But unlike \cite{melflow}, we add the speech embedding as a global condition to the WaveNet2D as follows:
\begin{equation}
    (\sigma, \mu)=WaveNet2D(z_a^{i-1}, \omega).
\end{equation}
The MelFlow operation is illustrated in Figure. \ref{melflow_arch}.

Once the MelFlow model is trained, we can estimate the mutual information upperbound as follows:
\begin{equation}
    \begin{aligned}
        L_{redundancy}=&E_{p(X,\omega)}[\log{p_X(X|\omega)}]\\
        &-E_{p(X)p(\omega)}[\log{p_X(X|\omega)}],
    \end{aligned}
    \label{redund_loss}
\end{equation}
where $\log{p_X(X|\omega)}$ is the conditional log-likelihood estimated using the MelFlow.

\subsection{Training strategy}

\begin{figure}
	\centering
	\begin{subfigure}[b]{\linewidth}
	\includegraphics[width=0.95\linewidth]{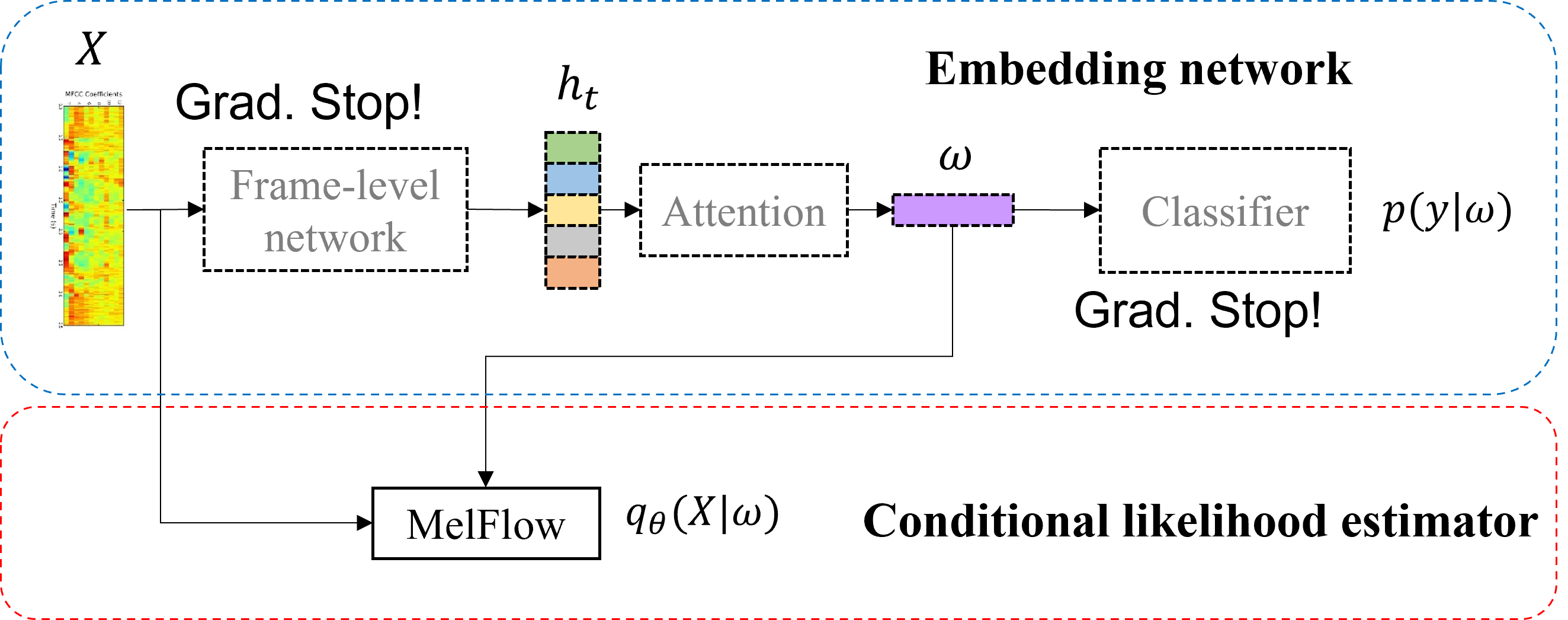}
	\caption{The MelFlow model is trained given the data and embedding pairs.}
	\end{subfigure}
	\begin{subfigure}[b]{\linewidth}
	\includegraphics[width=0.95\linewidth]{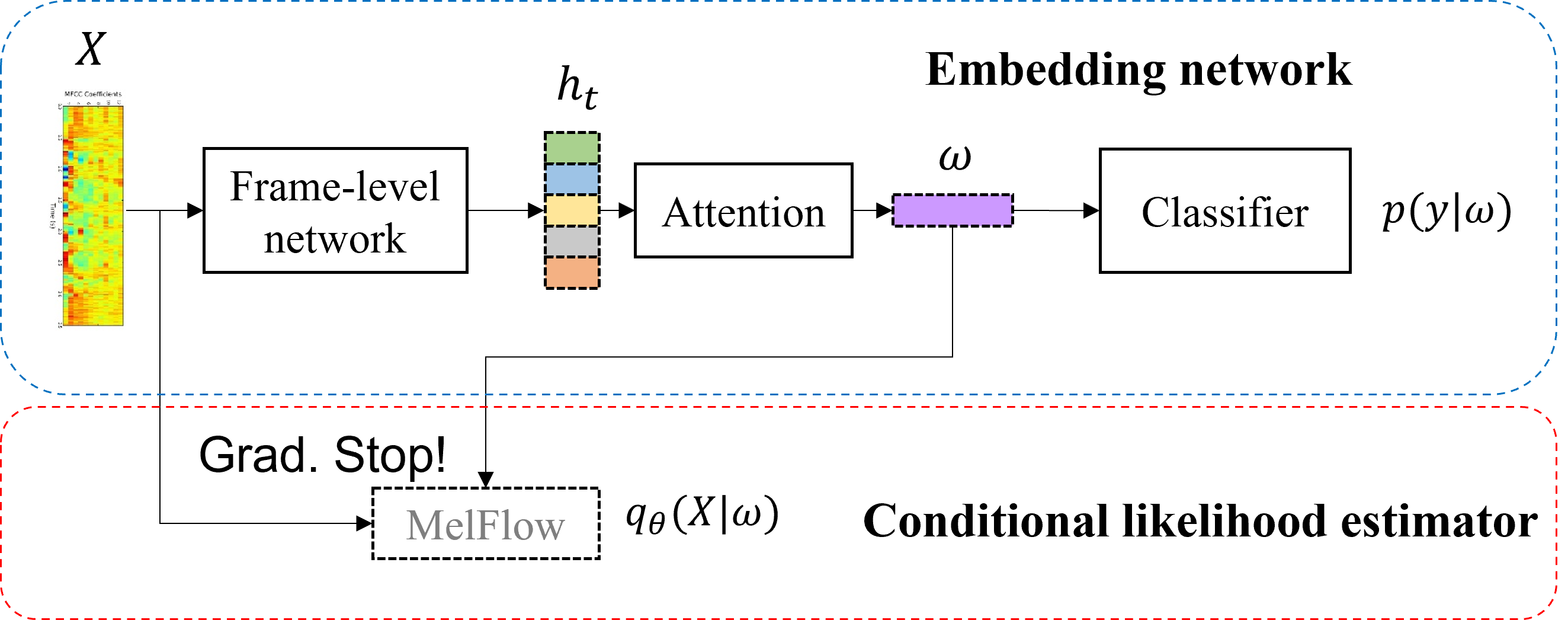}
	\caption{The embedding network is trained along with the classifier according to the information bottleneck loss $L_{IB}=-L_{xent}+{\beta}L_{redundancy}$, where $L_{redundancy}$ is computed using the conditional likelihood computed from the MelFlow.}
	\end{subfigure}
	\caption{The two step training process of the Flow-ER framework.}
	\label{flower_train}
\end{figure}

\begin{algorithm}[t]
\SetAlgoLined
\KwInput{Training set $X$, target label $Y$, embedding network with parameters $\Theta$, MelFlow network with parameters $\Phi$, information bottleneck coefficient $\beta$ and maximum epoch number $epoch^{max}$.}{}
 Initialize $\Phi$ and $\Theta$\;
 \While{$epoch < epoch^{max}$}{
  \eIf{$epoch = 0$}{
   $L_{xent} \leftarrow $ compute discriminative loss with $(X, Y, \Theta)$\;
   Optimize $\Theta$ with $L_{xent}$\;
   }{
   $\Omega \leftarrow $ extract embeddings with $(X, \Theta)$\;
   $\log p_X \leftarrow $ compute log-likelihood loss with $(X, \Omega, \Phi)$\;
   Optimize $\Phi$ with $\log p_X$\;
   $L_{xent} \leftarrow $ compute discriminative loss with $(X, Y, \Theta)$\;
   $L_{redundancy} \leftarrow $ compute redundancy loss with $(X, \Omega, \Phi$)\;
   Optimize $\Theta$ with $L_{IB}=-L_{xent}+{\beta}L_{redundancy}$\;
  }
 }
 \caption{Training steps for the Flow-ER framework}
 \label{algo_flower}
\end{algorithm}

In the proposed Flow-ER framework, the embedding network is trained according to the information bottleneck scheme, where the mutual information between the embedding $\omega$ and the label $y$ is maximized while the mutual information between $\omega$ and the input representation $X$ is minimized.
To accomplish, we optimize the network with the following objective function, which incorporates Eq. \ref{xent_mi} and \ref{redund_loss}:
\begin{equation}
    L_{IB}=-L_{xent}+{\beta}L_{redundancy},
\end{equation}
where $\beta$ is a predefined coefficient.

However, whenever the network parameters are updated, the distribution of the embedding vectors will change as well.
Therefore, in order to compute $L_{redundancy}$, the MelFlow should be updated with the new embedding vectors.
Thus as depicted in Figure \ref{flower_train} and Algorithm \ref{algo_flower}, we propose to train the embedding network and the MelFlow network in a competitive fashion, similar to the GAN training strategy.
The Flow-ER training is done in a 2-stage process: embedding network update and MelFlow update.
In the embedding network update phase, we freeze the MelFlow parameters and estimate the conditional likelihoods to compute $L_{redundacny}$. 
Then the embedding network and classification network parameters are updated through $L_{IB}=-L_{xent}+{\beta}L_{redundancy}$.
In the MelFlow update phase, the embedding network parameters are frozen and the embeddings are extracted.
Given the training data and their corresponding embeddings, the MelFlow is updated via likelihood maximization.

\section{Experiments}

To validate the impact of the proposed Flow-ER strategy in speech representation learning, we have conducted experiments in several speech processing tasks: speaker verification, voice anti-spoofing, and language identification.
In all three tasks, it is essential to maximize the discriminability in terms of the target class, while minimizing the information on the nuisance attributes.

\subsection{Experimental Setup}

\subsubsection{Speaker verification experimental setup}
In our speaker verification experiments, we have used the \textit{development} subset of the VoxCeleb2 dataset \cite{vox2}, consisting of 1,092,009 utterances collected from 5,994 speakers.
The evaluation was performed according to the original VoxCeleb1 trial list \cite{vox1}, which consists of 4,874 utterances spoken by 40 speakers.

The acoustic features used in the experiments were 40-dimensional MFCCs extracted at every 10 ms, using a 25 ms Hamming window via Kaldi toolkit \cite{kaldi}.

We have experimented with the ECAPA-TDNN \cite{ecapa}, an architecture that achieved state-of-the-art performance in text-independent speaker recognition. The ECAPA-TDNN uses squeeze-and-excitation as in the SE-ResNet, but also employs channel- and context-dependent statistics pooling and multi-layer aggregation.

The embedding networks are trained with segments consisting of 180 frames, using the ADAM optimization technique \cite{adam}.
The AAMSoftmax objective was used for training the embedding networks, and the experimented networks were implemented via PyTorch \cite{pytorch}, based on the voxceleb-trainer open-source project \cite{resnet1}\footnote{\url{https://github.com/joonson/voxceleb_trainer}}. 
The networks were trained with initial learning rate 0.001 decayed with ratio 0.95 for 150 epochs, and the models from the best performing checkpoint were selected.
The batch size for training was set to be 200.
Cosine similarity was used for computing the verification scores in the experiments.

\subsubsection{Voice anti-spoofing experimental setup}
In our anti-spoofing experiments, we have used the \textit{training} subset the ASVspoof 2019 challenge dataset was used for optimizing the systems, which provides a common framework with a standard corpus for conducting spoofing detection research on logical access (LA) attacks.
The LA dataset includes bonafide and spoof speech signals generated using various state-of-the-art voice conversion and speech synthesis algorithms. 
The evaluation was performed on the \textit{evaluation} subset, which consists seen and unseen test sets in terms of spoofing attacks. 
For more details about the corpora, the interested readers are referred to \cite{asvspoof2019}.

The acoustic feature used in the experiments was 60-dimensional (including the delta and double delta coefficients) linear frequency cepstral coefficients (LFCC) extracted using 25ms analysis window over a frame shift of 10ms.

We have experimented with the SE-ResNet-18, a variant of the ResNet-18, where a squeeze-and-excitation (SE) block \cite{se} is applied at the end of each non-identity branch of residual block to significantly decrease the computational cost of the system.

For training the experimented systems, the OCSoftmax objective function and balanced mini-batches of size 64 samples were used.
The ADAM optimizer was used with initial learning rate of 0.0003 and exponential learning rate decay with rate of 0.5 was applied \cite{monteiro2020generalized}.

\subsubsection{Language identification experimental setup}

In our language identification experiments, we have used the \textit{training} subset of the OLR2021 Challenge dataset \cite{olr21}, which consists of recordings from 13 languages.
The evaluation was performed on the \textit{progress} set of the OLR2021 Challenge dataset, where the evaluation metric was $C_{avg}$ and EER.

The acoustic features used in the experiments were 40-dimensional MFCCs and 3-dimensional pitch extracted at every 10 ms, using a 25 ms Hamming window via Kaldi toolkit \cite{kaldi}.

We have experimented with the ECAPA-TDNN \cite{ecapa} architecture as in the speaker verification experiment.
The embedding networks are trained with segments consisting of 180 frames, using the ADAM optimization technique \cite{adam}.
Analogous to the speaker verification experiments, the AAMSoftmax objective function was used for training the embedding systems.

\subsection{Results}

\begin{figure}
	\centering
	\begin{subfigure}[b]{0.49\linewidth}
	\includegraphics[width=\linewidth]{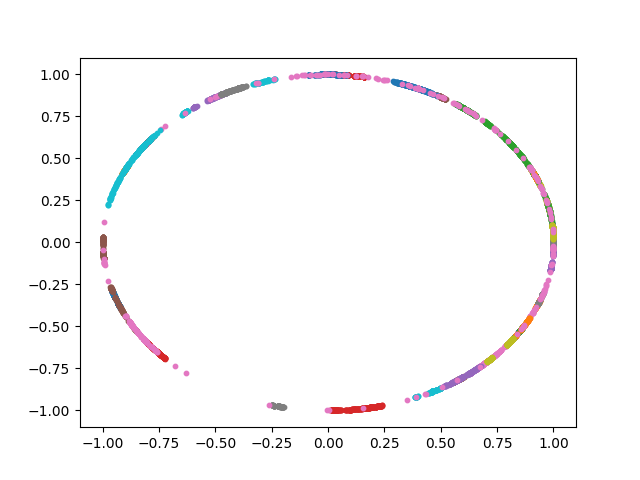}
	\caption{System trained with no regularization.}
	\end{subfigure}
	\begin{subfigure}[b]{0.49\linewidth}
	\includegraphics[width=\linewidth]{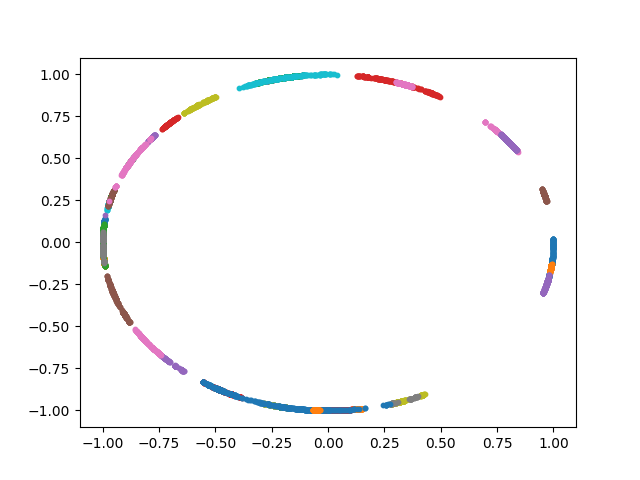}
	\caption{System trained with Flow-ER.}
	\end{subfigure}
	\caption{Normalized T-SNE plot of the speaker embeddings extracted from systems trained with and without Flow-ER. Different colors indicate distinct speakers.}
	\label{tsne_asv_flower}
\end{figure}

\begin{figure}
	\centering
	\begin{subfigure}[b]{0.49\linewidth}
	\includegraphics[width=\linewidth]{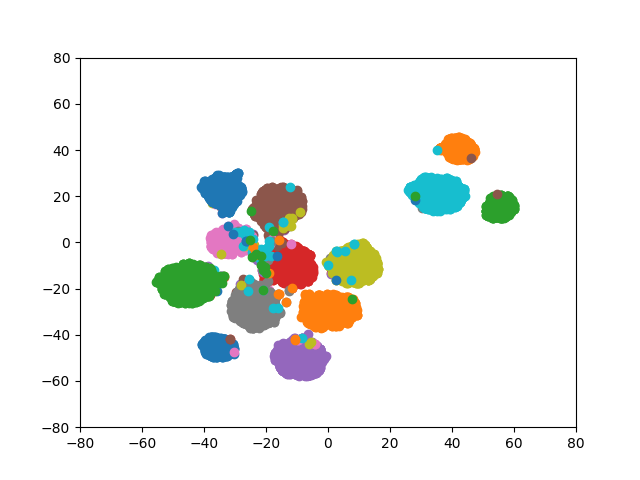}
	\caption{System trained with no regularization.}
	\end{subfigure}
	\begin{subfigure}[b]{0.49\linewidth}
	\includegraphics[width=\linewidth]{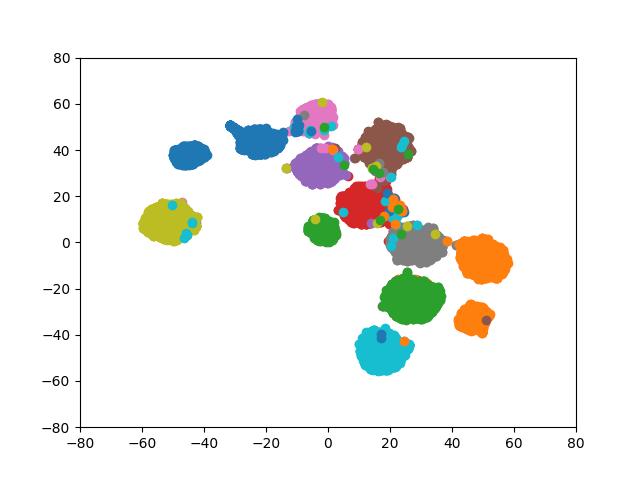}
	\caption{System trained with Flow-ER.}
	\end{subfigure}
	\caption{T-SNE plot of the language embeddings extracted from systems trained with and without Flow-ER. Different colors indicate distinct languages.}
	\label{tsne_flower}
\end{figure}

\begin{table*}[]
\centering
\caption{The performance comparison between baseline embedding system and system trained with the proposed Flow-ER on different tasks (i.e., speaker verification, anti-spoofing, language identification).}
\begin{tabular}{c|c|cc|cc|}
\cline{2-6}
\multirow{2}{*}{}                             & Speaker verification & \multicolumn{2}{c|}{Anti-spoofing}            & \multicolumn{2}{c|}{Language identification}      \\ \cline{2-6} 
                                              & EER {[}\%{]}         & \multicolumn{1}{c|}{EER {[}\%{]}} & min t-DCF & \multicolumn{1}{c|}{EER {[}\%{]}} & min $C_{avg}$ \\ \hline
\multicolumn{1}{|c|}{No regularization}       & 1.8240               & \multicolumn{1}{c|}{3.0589}       & 0.0718    & \multicolumn{1}{c|}{8.0940}       & 0.0671        \\ \hline
\multicolumn{1}{|c|}{\textbf{Flow-ER ($\beta=0.001$)}} & \textbf{1.7391}               & \multicolumn{1}{c|}{\textbf{2.8029}}       & \textbf{0.0619}    & \multicolumn{1}{c|}{\textbf{7.4370}}       & \textbf{0.0639}        \\ \hline
\end{tabular}
\label{results}
\end{table*}

\subsubsection{Analysis of the embeddings on the embedding space}

Figure \ref{tsne_asv_flower} depicts the normalized T-SNE plot of the speaker embeddings extracted from systems trained with and without the proposed Flow-ER.
From the embeddings extracted from the system trained without embedding regularization, we could observe numerous overlaps between samples from different speaker identities, which may be caused by the non-speaker attributes.
Meanwhile from the T-SNE plot of the Flow-ER embeddings, the overlapping is significantly alleviated.

Moreover, Figure \ref{tsne_flower} shows the T-SNE plot of the language embeddings extracted from systems trained with and without the proposed Flow-ER.
From the embeddings extracted using the conventional method, it could be seen that some clusters are far away from each other even if they have the same class identity.
Such variability may be attributed to the nuisance attributes, such as gender or speaker of the utterance.
On the other hand, in the embeddings trained with the proposed Flow-ER, the clusters with the same class identity are relatively much closer to each other, and the general distribution of the embeddings is more spread out than the conventional embeddings.
From these observations, we could assume that the proposed Flow-ER can help the embeddings to have better discriminability by disentangling the nuisance attributes from them.

\subsubsection{Performance of Flow-ER in different down-stream tasks}

The experimental results of the systems trained only with discriminative loss and ones trained with the proposed Flow-ER strategy are depicted in Table \ref{results}.
As shown in the results, it could be observed that the proposed Flow-ER can improve the performance in all three tasks.
Especially in anti-spoofing, the Flow-ER system outperformed the baseline with a relative improvement of 13.79\% in terms of min t-DCF.
These results tell us that the proposed Flow-ER is able to effectively improve the performance of various down-stream tasks.

\section{Conclusion}
In this paper, we proposed a novel approach, which we call Flow-ER, to disentangle the nuisance information from the speech embedding vector.
The proposed method exploits the information bottleneck framework, where the embedding network is trained to have maximum information on the main-task, while suppressing the information of the unwanted attributes.
To incorporate the information bottleneck scheme into the embedding network training process, the mutual information is estimated using the main task classifier and an auxiliary normalizing flow network.

In order to evaluate the proposed Flow-ER strategy, we have conducted several experiments on different down-stream tasks, including speaker verification, antispoofing, and language identification.
Our results showed that the Flow-ER can improve the performance in all the experimented tasks, which may be attributed to its capability in disentangling the nuisance information from the embeddings.
Especially in anti-spoofing, the Flow-ER system outperformed the baseline with a relative improvement of 13.79\% in terms of min t-DCF.

In our future study, we will be investigating the potential of the Flow-ER strategy furthermore by applying it to more diverse tasks, including image classification and anomaly detection.
Moreover, as the sensitivity of the mutual information estimation in the proposed Flow-ER will highly vary depending on the accuracy of the estimated conditional likelihood, the choice of the conditional likelihood estimator may be crucial for the disentanglement performance.
Therefore we will also focus our research on finding the optimal normalizing flow model for the conditional likelihood estimator.

\bibliographystyle{named}
\bibliography{main.bib}

\end{document}